\documentclass[12pt]{article}
\usepackage[utf8]{inputenc}
\usepackage[a4paper]{geometry}
\usepackage{natbib}
\usepackage{graphicx}
\usepackage{amsmath,amssymb,amsthm}
\usepackage{color}
\usepackage{hyperref}
\usepackage[pagewise]{lineno}
\usepackage{etoolbox}

\newtheorem{corollary}{Corollary}[section]
\newtheorem{definition}{Definition}[section]
\newtheorem{theorem}{Theorem}[section]
\newtheorem{proposition}{Proposition}[section]
\newtheorem{lemma}{Lemma}[section]
\theoremstyle{remark}

\theoremstyle{remark}

\newtheoremstyle{red}
  {\topsep}
  {\topsep}
  {}
  {0pt}
  {\bfseries}
  {}
  {.5em}
  {\ifstrequal{#3}{red}{\color{red}}{}\thmname{#1}\thmnumber{ #2:}}
\theoremstyle{red}
\newtheorem{algorithm}{Algorithm}[section]

\title{A general Bayesian bootstrap for censored data based on the beta-Stacy process}
\author{Andrea Arf\`e\footnote{Department of Epidemiology and Biostatistics, Memorial Sloan Kettering Cancer Center, New York, NY 10017, United States. Website: \url{andreaarfe.wordpress.com}. E-mail: \url{arfea@mskcc.org}.}, Pietro Muliere\footnote{Department of Decision Sciences, Bocconi University, 20136 Milan, Italy. E-mail: \url{pietro.muliere@unibocconi.it}}}

\begin{document}
\maketitle

\begin{abstract}
We introduce a novel procedure to perform Bayesian non-parametric inference with right-censored data, the \emph{beta-Stacy bootstrap}. This approximates the posterior law of summaries of the survival distribution (e.g. the mean survival time). More precisely, our procedure approximates the joint posterior law of functionals of the beta-Stacy process, a non-parametric process prior that generalizes the Dirichlet process and that is widely used in survival analysis. The beta-Stacy bootstrap generalizes and unifies other common Bayesian bootstraps for complete or censored data based on non-parametric priors. It is defined by an exact sampling algorithm that does not require tuning of Markov Chain Monte Carlo steps. We illustrate the beta-Stacy bootstrap by analyzing survival data from a real clinical trial. 
\end{abstract}

 
\section{Introduction}

Survival data is often censored, hindering statistical inferences \citep{Kalbfleisch2002}. In this setting, the goal is often to perform inference on specific summaries $\phi(G)$ of the  cumulative distribution function $G(x)$ (defined for $x\geq 0$) that generated the observed survival times $Y_1,\ldots,Y_n$, e.g. the expected survival time, or the probability to survive past a landmark time-point. 

We introduce \emph{beta-Stacy bootstrap}, a new method to perform Bayesian non-parametric inference for functionals $\phi(G)$ of the distribution function $G$ using censored data. Specifically, the proposed approach generates approximate samples from the posterior law of $\phi(G)$ obtained by assuming that $G$ is a beta-Stacy process \citep{Walker1997}. This process defines a non-parametric prior law for distribution functions widely used with censored data (\citealp{walker1998}; \citealp{AlLabadi2013}; \citealp{Arfe2018a}). The beta-Stacy process extends the classical Dirichlet process of \citet{Ferguson1973} and it is conjugate to both complete and right-censored data \citep{Walker1997}. It is also strictly related to the \emph{beta process} of \citet{Hjort1990}: $G$ is a beta-Stacy process if and only if its cumulative hazard function is a beta process \citep{Walker1997}.

The proposed approach belongs to the family of \emph{Bayesian bootstrap} procedures pioneered by \citet{Rubin1981}. In addition to Rubin's, this family includes the \emph{proper Bayesian bootstrap} of \citet{Muliere1996}, the \emph{Bayesian bootstrap for censored data} of \citet{Lo1993}, and others \citep{Lo1991, Kim2003, Lyddon2019}. Similarly to Efron's classical bootstrap \citep{Efron1986}, Bayesian bootstraps repeatedly re-sample and/or re-weight the observed data to induce a probability distribution for $\phi(G)$. More precisely, Bayesian bootstraps generate approximate samples from the posterior distribution of $\phi(G)$ associated to some non-parametric prior for $G$ (for connections with Efron's frequentist procedure, see \citealp{Lo1987, Lo1991, Lo1993, Muliere1996}). Interest in these sampling algorithms has recently increased thanks to their scalability and computational simplicity---e.g. they do not require tuning of Markov Chain Monte Carlo steps \citep{Lyddon2018, Barrientos2020}.



We show that the beta-Stacy bootstrap generalizes other common Bayesian bootstrap procedures. These include those of \citet{Rubin1981} and \citet{Muliere1996}, which are at the core of other recent proposals \citep{Lyddon2019, Barrientos2020}, but cannot be applied in presence of censoring. They also include Lo's procedure (\citeyear{Lo1993}), which can incorporate censored observations, but cannot incorporate prior information on the functional form of $G$. We characterize each of these methods as a special or limiting case of the beta-Stacy bootstrap (c.f. Figure \ref{fig:1}), which, in comparison, can be applied with censored data and allows to incorporate prior information on the data-generating distribution. 

\begin{figure}[tbh]
\centering
\includegraphics[scale=0.6, trim={6cm 4cm 6cm 4cm}, clip]{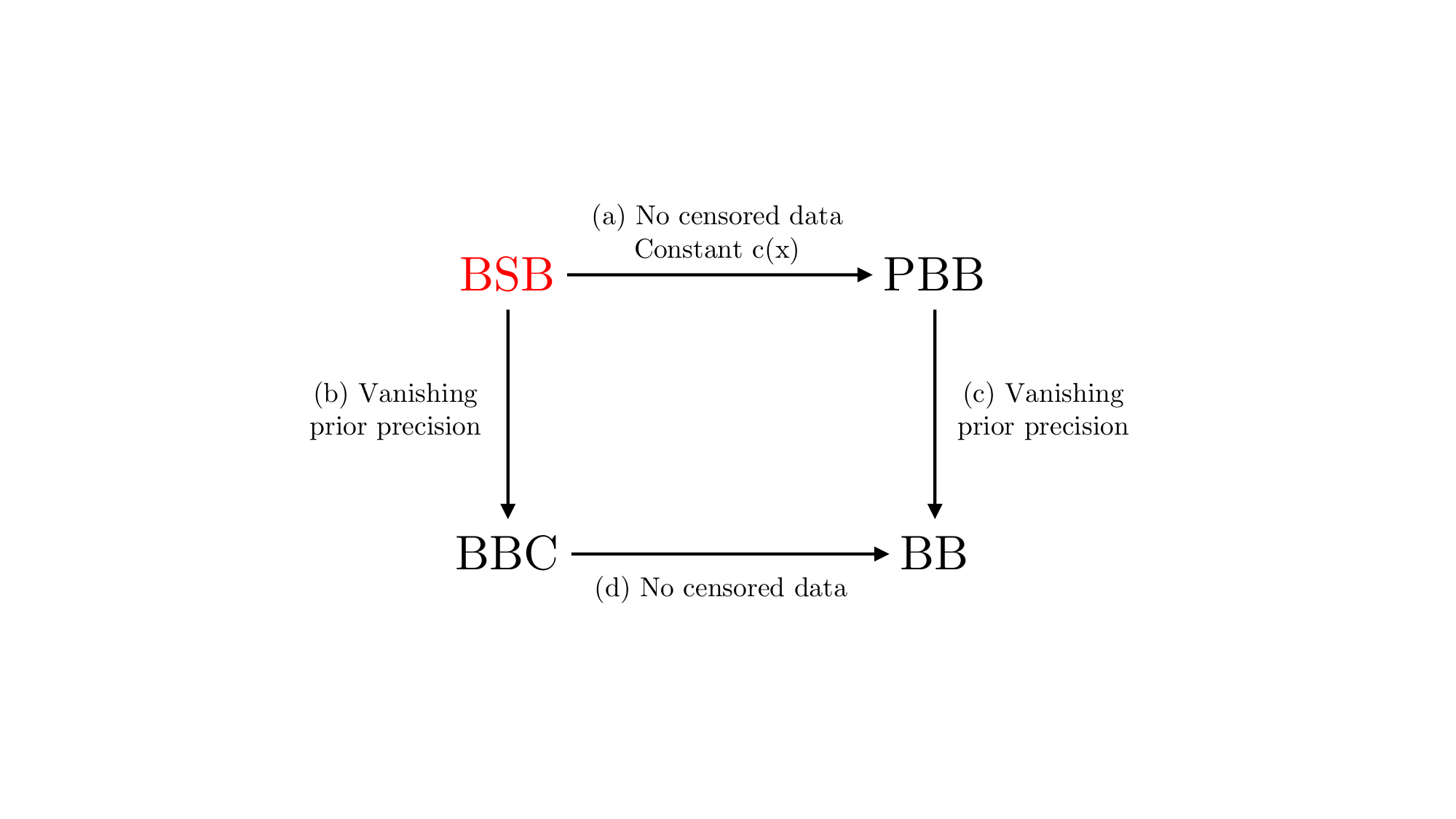}
\caption{Relations between different Bayesian bootstraps: BSB (in red), beta-Stacy Bootstrap (c.f. Section \ref{sec:bsb}); PBB, Proper Bayesian Bootstrap (\citealp{Muliere1996}); BBC, Bayesian Boostrap for Censored data (\citealp{Lo1993}); BB, classical Bayesian Boostrap (\citealp{Rubin1981}). The prior precision of the BSB is controlled by a function $c(x)$, while that of the PBB is controlled by a constant $k$. (a) the BSB and PBB coincide when there is no censoring and $c(x)=k$; (b) the BSB reduces to the BBC if $c(x)\rightarrow 0$ for every $x$; (c) the PBB reduces to the BB if $k\rightarrow 0$; (d) the BCC and BB coincide when there is no censoring. See Section \ref{sec:boot} for details.}
\label{fig:1}
\end{figure}

We note that, when $G$ has a beta-Stacy prior distribution, posterior inferences for $\phi(G)$ could also be based on algorithms for the simulation of L\'evy processes (c.f.  \citealp{Damien1995}, \citealp{walker1998}, \citealp{Ferguson1972}, and \citealp{Wolpert1998}; see also \citealp[Section 13.3.3]{Ghosal2017} and \citealp{DeBlasi2014} for a reviews and applications to the beta-Stacy process). With these methods, it is possible to generate approximate samples from the posterior law of $G$, and so also from the posterior law of $\phi(G)$. However, some algorithms (e.g. \citealp{Damien1995, walker1998}) can only generate approximate sample paths $\{G(x): x\in [0,T]\}$ over some bounded interval $[0,T]$. Hence, they may be difficult to apply to summaries $\phi(G)$ that depend on all values of $G$, such as the expected survival time. These cases are not problematic for the beta-Stacy bootstrap. Other approaches (e.g. \citealp{Ferguson1972, Wolpert1998}) can approximately sample full paths $\{G(x): x\in [0,+\infty)\}$ from the posterior law of $G$, but they are computationally more complicated than the beta-Stacy bootstrap (e.g. they may require auxiliary algorithms to sample from unnormalized distributions). 
 
The rest of the paper is structured as follows. In Section 2, we introduce notations and assumptions used throughout the manuscript. In Section 3, we review the definition and properties of the beta-Stacy process. In Section 4, we introduce the beta-Stacy bootstrap and study its approximation properties (most technical proofs are provided in Appendix). In Section 5, we describe the connections of the beta-Stacy bootstrap with other Bayesian bootstrap algorithms. In Section 6, we briefly describe a generalization of the beta-Stacy bootstrap to the $k$-sample setting. In Section 7, we describe a computational approach for implementing the beta-Stacy bootstrap. Using data from a clinical trial in hepatology \citep{Dickson1989}, in Section 8 we illustrate the beta-Stacy bootstrap and contrast it to an algorithm that generates approximate beta-Stacy sample paths. We describe this algorithm in the Supplementary Material, where we also report on additional comparative simulation studies (c.f. Section 8.4). Finally, Section 9 provides concluding remarks and discusses potential venues for future research. Code to replicate our analyses is available online at \url{https://github.com/andreaarfe/} or by request to the first author.

\section{Basic notations and assumptions}\label{sec:1}

If $Z:[0,+\infty)\rightarrow \mathbb{R}$ is a non-decreasing, right-continuous function with left-hand limits, we let $\overline{Z}(x)=1-F(x)$ and $\Delta Z(x)=Z(x)-Z(x-)$ for every $x\geq 0$ (where $Z(0-)=0$). We also identify $Z$ with its induced measure, writing $Zf=\int h(x)dZ(x)$ for any function $h(x)$, and $Z(S)=\int I\{x\in S\}dZ(x)$ for any $S\subseteq [0, +\infty)$. A function $h(x)$ is $Z$-\emph{integrable} if $Z|h|<+\infty$. We will denote with $D_h\subseteq [0,+\infty)$ the set of discontinuity points of $h$, and say that $h$ is $Z$-\emph{almost everywhere continuous} if $Z(D_h)=0$ (this is true when $h$ is continuous, and it implies that $h$ must be continuous at every atom of $Z$). If $Z$ is random, then its distribution is fully characterized by its \emph{Laplace functional}, i.e. the map $h(x)\mapsto \mathbb{E}[\exp(-Zh)]$, where $h(x)$ is any non-negative function \citep[Chapter 2]{Kallenberg2017}. 

We assume that $T_1,\ldots,T_n$ are independent, survival times, each with the same cumulative distribution function $G$ (with $G(0)=0$). In survival analysis applications, it is common for $T_1,\ldots,T_n$ to be \emph{(right) censored}. In such cases, the observed dataset is formed by $Y_1=(T_1^c,\delta_1)$, $\ldots$, $Y_n=(T_n^c,\delta_n)$, where, for each $i=1,\ldots,n$, $T_i^c = \min(T_i, C_i)$ is the censored version of $T_i$, $C_i$ its censoring time, and $\delta_i=I\{T_i\leq C_i\}$ its censoring indicator. As common in this setting, we assume that censoring is \emph{independent}---i.e. that $C_1,\ldots, C_n$ are independent of $T_1,\ldots,T_n$ \citep[Section 3.2]{Kalbfleisch2002}---and \emph{ignorable}---which essentially means that $C_1, \ldots, C_n$ can be treated as known constants when computing posterior distributions \citep{Heitjan1991, Heitjan1993}. We will also use the same notations when there is no censoring, in which case we simply define $T_i^c=T_i$ and $\delta_i=1$ for every $i=1,\ldots,n$. To refer to either these situations, we will simply say that the (potentially censored) survival times $Y_1,\ldots,Y_n$ are generated by $G$. 

Let $Y_1,\ldots,Y_n$ be (possibly censored) survival times generated from a distribution function $G$. Our aim is to make inferences on $\phi(G)=f(Gh_1,\ldots,Gh_k)$, a summary of $G$ defined by the real-valued functions $f(x_1,\ldots,x_k)$ and $h_1(x)$, $\ldots$, $h_k(x)$ (later we consider vectors of such summaries). Examples include the mean ($h_1(x)=x$, $f(x_1)=x_1$), the variance ($h_1(x)=x^2$, $h_2(x)=x$, $f(x_1,x_2)=x_1-x_2^2$), or the restricted mean survival time ($h_1(x)=\min(x,\tau)$, $f(x_1)=x_1$, $\tau>0$; \citealp{Royston2013}). From the Bayesian non-parametric perspective, any inference on $\phi(G)$ can be accomplished first by assuming that $G$ is distributed according to some prior process, then computing or approximating the posterior distribution of $G$, and so of $\phi(G)$, conditional on the observed data $Y_1,\ldots,Y_n$.

Let $k>0$ and $F$ be a distribution function over $[0, +\infty)$. We say that $G$ is a \emph{Dirichlet process} $\textrm{DP}(k,F)$ and write $G\sim \textrm{DP}(k,F)$ if, for all $0=x_0<x_1<\ldots<x_k<x_{k+1}=+\infty$, $(G(x_1)-G(x_0),\ldots,G(x_{k+1})-G(x_k))$ has Dirichlet distribution $\textrm{Dir}(\alpha_1,\alpha_2,\ldots,\alpha_{k+1})$, where $\alpha_i=k(F(x_i)-F(x_{i-1}))$. If $G\sim \textrm{DP}(k,F)$ and there is no censoring, the posterior law of $G$ conditional on $Y_1,\ldots,Y_n$ is $\textrm{DP}(k+n,F^*)$, with $F^*(x)=\frac{k}{k+n}F(x)+\frac{1}{k+n}\sum_{i=1}^nI\{T_i\leq x\}$ \citep{Ferguson1973}. However, if any $T_i$ is censored (i.e. $\delta_i=0$), then this posterior distribution is not a Dirichlet process anymore \citep{Ferguson1979, Walker1997}. In contrast, the beta-Stacy process is conjugate with respect to censored data, allowing simple posterior computations \citep{Walker1997}. (Later, we will also discuss the use prior processes different that the beta-Stacy in the considered setting.)

\section{The beta-Stacy process prior}\label{sec:bs}

The beta-Stacy process is the law of a random cumulative distribution function $G(x)$ with support in $[0,+\infty)$ \citep{Walker1997}. It is a \emph{neutral-to-the-right}, a type of non-parametric priors widely used with censored data \citep{Doksum1974, Ferguson1979}. This means that if $Z(x)=-\log(1-G(x))$, then the increments $Z(t_1)-Z(t_0)$, $Z(t_2)-Z(t_1)$, $\ldots$, $Z(t_k)-Z(t_{k-1})$ are independent for every $0=t_0<t_1<\ldots<t_k$ (\citealp[Chapter 13]{Ghosal2017}). 

Let $F(x)$ be a cumulative distribution function with $F(0)=0$ and jumps at locations $x_1<x_2<\ldots$ (so $\Delta F(x_j)>0$ for every $x_j$). Also let $c(x)>0$ for every $x\geq 0$. 
\begin{definition}[\citealp{Walker1997}] The cumulative distribution function $G$ is \emph{beta-Stacy} process $\textrm{BS}(c,F)$ if the Laplace functional of $Z$ satisfies
\begin{equation}\label{eqn:bslapl}
-\log\mathbb{E}[\exp(-Zh)] = \int_0^{+\infty}\int_0^{+\infty} (1-e^{-uh(x)}) \rho(x,u) dF(x)du
\end{equation}
for every $h(x)\geq 0$, where 
\begin{equation}\label{eqn:rho}
\rho(x,u) = \frac{1}{1-e^{-u}} c(x) \exp\left(-uc(x)\overline{F}(x)\right)r\left(uc(x)\Delta F(x)\right)
\end{equation}
and $r(u)=(1-e^{-u})/u$ for $u>0$, $r(0)=1$.
\end{definition}

The sample paths of $G(x)$ are discrete, as $Z(x)$ can only increase by an at most countable number of jumps \citep{Walker1997}. A jump always occur at each atom $x_j$ of $F(x)$; its size is $\Delta G(x_j)=U_j \prod_{x_i < x}(1-U_i)$ for independent $U_j=1-\exp(-\Delta Z(x_j))\sim \textrm{Beta}\left(c(x_j)\Delta F(x_j), c(x_j)\overline{F}(x_j)\right)$.  When $F$ is discrete, $G$ can only jump at each $x_j$, so $\overline{G}(x)=\prod_{x_j\leq x}(1-U_j)$ for $x>0$. Otherwise, some jumps also occur at random positions. Their locations and sizes are determined by the $x$- and $u$-coordinates of the points $(x,u)$ of a non-homogeneous Poisson process on $(0,+\infty)^2$; this is independent of each $U_j$ and has intensity measure $(1-e^{-u})^{-1} c(x) \exp\left(-uc(x)\overline{F}(x)\right)dF_c(x)du,$ where $F_c(x)=F(x)-\sum_{x_j\leq x}\Delta F(x_j)$ is the continuous part of $F$. 

If $G \sim \textrm{BS}(c,F)$, then $dG(x)/\overline{G}(x-) \sim \textrm{Beta}(c(x)dF(x),c(x)\overline{F}(x))$, infinitesimally speaking \citep{Walker1997}. Hence, $\mathbb{E}[dG(x)]=dF(x)$, and so $\mathbb{E}[G(x)]=F(x)$, for all $x>0$. Moreover, the variance of $dG(x)$ is a decreasing function of $c(x)$, with $\textrm{Var}(dG(x))\rightarrow 0$ as $c(x)\rightarrow +\infty$. The function $c(x)$ thus controls the dispersion of the distribution $\textrm{BS}(c,F)$ around its mean $F$. Throughout, we will assume that i) $F(x)<1$ for all $x>0$ and ii) $\epsilon\leq c(x)\leq \epsilon^{-1}$ for all $x>0$ and some $\epsilon\in (0,1)$. The former condition implies that $Z(x)$ has finite value (and so $G(x)<1$) with probability 1 for every $x>0$. The latter instead rule out extreme cases in which $dG(x)$ has null or arbitrarily large variance for some $x>0$. (Both are technical requirements needed to prove Lemma \ref{lemma:weakconv} in the Appendix.)   


As previously mentioned, the classical Dirichlet process is a special case of the beta-Stacy process. In fact, \citet{Walker1997} show that if $c(x)=k$ for all $x>0$, then $\textrm{BS}(c,F)=\textrm{DP}(k,F)$. Contrary to the Dirichlet process, however, the beta-Stacy process is conjugate with respect to right-censored data. Specifically, assume that i) $Y_1,\ldots,Y_n$ are generated by $G\sim \textrm{BS}(c,F)$; ii) $N(x)=\sum_{i=1}^nI\{T_i^c\leq x, \delta_i=1\}$ is the number of uncensored survival times less or equal than $x\geq 0$; and iii) $M(x)=\sum_{i=1}^n I\{T_i^c \geq x\}$ for all $x\geq 0$. Then we have the following result:
\begin{theorem}[Theorem 4, \citealt{Walker1997}]\label{thm:bspost}
The posterior distribution of $G$ conditional on $Y_1,\ldots,Y_n$ is the beta-Stacy process $\textrm{BS}(c^*,F^*)$, where
\begin{equation}\label{eqn:fpost}
F^*(x) = 1 - \prod_{t\in [0,x]} \left[ 1-\frac{c(t)dF(t)+dN(t)}{c(t)\overline{F}(t-)+M(t)} \right],
\end{equation}
\begin{equation}\label{eqn:cpost}
c^*(x) = \frac{c(x)\overline{F}(x-)+M(x)-\Delta N(x)}{1-F^*(x-)},
\end{equation}
and $\prod_{t\in [0,x]}$ is the product integral operator of \citet{Gill1990}.
\end{theorem}

The posterior mean $F^*(x)=\mathbb{E}[G(x)|Y_1,\ldots,Y_n]$ from Equation \ref{eqn:fpost} converges to $\widehat{G}(x)=1-\prod_{t\in [0,x]} \left[1-dN(t)/M(t) \right]$, the standard Kaplan-Meier estimator of the distribution function, as $c(x)\rightarrow 0$ for all $x>0$ \citep{Walker1997}.

In practice, $F^*(x)$ can be computed as $F^*(x)=1-(1-F_d^*(x) )(1-F_c^*(x))$, where, respectively, $F_d^*$ and $F_c^*$ are the following discrete and continuous distribution functions \citep{Gill1990}. First, 
\begin{equation}\label{eqn:disc}
F^*_d(x) = 1- \prod \left[1-\frac{c(t)\Delta F(t)+ \Delta N(t)}{(c(t)\overline{F}(t-)+M(t)}\right],
\end{equation}
where the product ranges over all positive $t\leq x$ such that $\Delta F(t) + \Delta N(t)>0$ (which are at most countable). Second, 
\begin{equation}\label{eqn:cont}
F^*_c(x) =1-  \exp\left(- \int_0^x \frac{c(t)dF_c(t)}{c(t)(1-{F}(t-))+M(t)}\right),
\end{equation}
where $F_c(x) = F(x)-\sum_{x_j\leq x}\Delta F(x_j)$ is $F(x)$ with the discontinuities removed.


\section{The beta-Stacy bootstrap}\label{sec:bsb}

We now introduce the beta-Stacy bootstrap. Let $Y_1,\ldots,Y_n$ be (possibly censored) survival times generated by $G\sim\textrm{BS}(c,F)$. The proposed procedure approximately samples from the law of $\phi(G)=f(Gh_1,\ldots,Gh_k)$ conditional on $Y_1,\ldots,Y_n$. Better, it samples from an approximation to the law of $\phi(G^*)$, where $G^*\sim\textrm{BS}(c^*,F^*)$ and $F^*$, $c^*$ are from Equations \ref{eqn:fpost} and \ref{eqn:cpost}. 
 
\begin{algorithm}\label{def:bsb}
The beta-Stacy bootstrap is defined by the following steps:
\begin{enumerate}
\item Sample $X_1,\ldots,X_m$ from $F^*$ and determine the corresponding number $D$ of distinct values $X_{1,m}<\cdots<X_{D,m}$ (later we describe how to implement this step in practice and provide guidance on how to choose $m$).
\item Compute $\alpha_i=c^*(X_{i,m})\Delta F_m(X_{i,m})$, $\beta_i=c^*(X_{i,m})\overline{F}_m(X_{i,m})$ for every $i=1,\ldots,D$, where $F_m(x)=\sum_{i=1}^m I\{X_i\leq x\}/m$ is the empirical distribution function of $X_1,\ldots,X_m$.
\item For all $i=1\ldots,D$, generate $U_i\sim\textrm{Beta}(\alpha_i,\beta_i)$ (with $U_D=1$, as $\beta_D=0$) and let $Z_i=U_i\prod_{j=1}^{i-1}(1-U_j)$.
\item Let $G_m(x)=\sum_{i=1}^{D}I\{X_{i,m}\leq x\}Z_i$ and compute $\phi(G_m)=f(G_m h_1$, $\ldots$, $G_mh_k)$, where $G_m h_j = \sum_{i=1}^{D}h(X_{i,m})Z_i$ for all $j=1,\ldots,k$.
\item Output $\phi(G_m)$ as an approximate sample from the distribution of $\phi(G^*)$.
\end{enumerate}
\end{algorithm}

We note that, in step 2 above, $\Delta F_m(X_{i,m})$ and $\overline{F}_m(X_{i,m})$ are just the proportions of values $X_1,\ldots,X_m$ that are equal to or stricter that $X_{i,m}$, respectively. We also note that the law of $G_m$ in step 4 is the mixture of the beta-Stacy process $\textrm{BS}(c^*,F_m)$ with mixing measure $\prod_{i=1}^m F^*(dx_i)$, the joint law of $X_1,\ldots,X_m$. This generalizes the \emph{Dirichlet-multinomial process}, which is a mixture of Dirichlet process with mean $F_m$  (\citealp{Ishwaran2002, Muliere2003a}).

Some of the $X_1,\ldots,X_m$ sampled in step 1 can be equal to one of the observed uncensored event times among $Y_1,\ldots,Y_n$. This is because every observed event time is an atom of $F^*$, as shown by Equation \ref{eqn:disc}. However, some the values $X_1,\ldots,X_m$ can also be new observations sampled from the support of the prior mean $F$ (e.g. these may come from the continuous component of $F^*$ in Equation \ref{eqn:cont}). This deviates from other Bayesian bootstrap procedures, which typically only incorporate observed data \citep{Rubin1981, Lo1993}.

The following result shows that, if $h$ is $F^*$-integrable (so that the posterior mean of $Gh$ exists finite) and $F^*$-almost everywhere continuous (a necessary technical condition to prove this result; c.f. the Appendix), then the law of $G_m h$ generated by Algorithm \ref{def:bsb} using data $Y_1,\ldots,Y_n$ approximates the posterior law of $G h$ conditional on $Y_1,\ldots,Y_n$ for large $m$. More precisely, it shows that $G_m h$ convergences in law to $Gh$ conditional on $Y_1,\ldots,Y_n$, i.e. $\mathbb{E}[H(G_mh)|Y_1,\ldots,Y_n] \rightarrow \mathbb{E}[H(Gh)|Y_1,\ldots,Y_n]$ as $m\rightarrow +\infty$ for any bounded continuous function $H$ (note that the sample size $n$ is fixed, and only the number of resamples $m$ varies).

\begin{proposition}\label{prop:gdmconv}
If $h:[0,+\infty)\rightarrow\mathbb{R}$ is $F^*$-integrable and $F^*$-almost everywhere continuous, then $G_m h\rightarrow G^*h$ in law for $m\rightarrow+\infty$ conditional on $Y_1,\ldots,Y_n$.
\end{proposition}
\begin{proof}
The proof is provided in the Appendix, as it relies on multiple lemmas. 
\end{proof}

The following corollary implies that the beta-Stacy bootstrap can also approximates the joint distribution of vectors of the form $(G^*h_1$, $\ldots$, $G^*h_k)$. This is useful to approximate the joint distribution of multiple summaries of $G$, e.g. the joint distribution of it's first $k$ moments ($h_j(x)=x^j$ for $j=1,\ldots,k$). 

\begin{corollary}\label{cor:1} Let $h_1,\ldots,h_k$ be $F^*$-integrable and $F^*$-everywhere continuous. Then, $(G_mh_1$, $\ldots$, $G_mh_k)\rightarrow(G^*h_1$, $\ldots$, $G^*h_k)$ in law conditional on $Y_1,\ldots,Y_n$ for $m\rightarrow+\infty$, i.e. $\mathbb{E}[H(G_mh_1, \ldots, G_mh_k)|Y_1,\ldots,Y_n]$ $\rightarrow$ $\mathbb{E}[H(G^*h_1, \ldots, G^*h_k)|$ $Y_1,\ldots,Y_n]$ as $m\rightarrow +\infty$ for any bounded continuous function $H$.
\end{corollary}
\begin{proof}
Take $\lambda_1,\ldots,\lambda_k\in\mathbb{R}$ and define $h^*=\lambda_1 h_1 + \cdots + \lambda_k h_k$. By Proposition \ref{prop:gdmconv}, $G_m h^*\rightarrow G^*h^*$ for $m\rightarrow +\infty$. This implies that the joint characteristic function of $(G_mh_1$, $\ldots$, $G_mh_k)$ converges to that of $(G^*h_1,\ldots,G^*h_k)$: $\mathbb{E}[\exp(i\{\lambda_1G_mh_1 + \cdots + \lambda_k G_mh_k\})|Y_1,\ldots,Y_n]$ $=$ $\mathbb{E}[\exp(iG_mh^*)|Y_1,\ldots,Y_n]$ $\rightarrow$ $\mathbb{E}[\exp(iG^*h^*)|Y_1,\ldots,Y_n]$ $=$ $\mathbb{E}[\exp(i\{\lambda_1G^*h_1 + \cdots + \lambda_k G^*h_k\})|Y_1,\ldots,Y_n]$ for $m\rightarrow +\infty$. \qedhere
\end{proof}

Consequently, for large $m$ the law of the sample $\phi(G_m)=f(G_m h_1,\ldots,G_m h_k)$ generated by the beta-Stacy bootstrap is approximately the same of $\phi(G^*)=f(G^*h_1,\ldots,G^*h_k)$. In fact, if $f$ is continuous (as is the case for all examples considered in this paper), then by Corollary \ref{cor:1} and the continuous mapping theorem it holds that  $\phi(G_m)\rightarrow\phi(G^*)$ in law as $m\rightarrow+\infty$. Hence, if $m$ is sufficiently large (e.g. $m\approx 1,000$; see Section \ref{sec:app}), by repeating steps 1-4 above independently, it is possible to generate an approximate sample of arbitrary size from the posterior law of $\phi(G)$. More generally, the joint law of $(\phi_1(G_m), \ldots, \phi_k(G_m))$ converges to that of $(\phi_1(G^*), \ldots, \phi_p(G^*))$, where $\phi_j(G)=f_j(Gh_1,\ldots,Gh_k)$ and $f_j(x_1,\ldots,k)$ is continuous or all $j=1,\ldots p$. Thus the beta-Stacy bootstrap can also be used to approximate the joint posterior law of vectors of functionals of $G$.

\section{Connection with other Bayesian bootstraps}\label{sec:boot}

The proposed procedure is a Bayesian analogue of Efron's classical bootstrap \citep{Efron1981}. When censoring is possible, the latter is based on repeated sampling from the Kaplan-Meier estimator $\widehat{G}$ \citep{Efron1986}. Similarly, the beta-Stacy bootstrap samples from $F^*$ (c.f. step 1 of Definition \ref{def:bsb}), the beta-Stacy posterior mean from Theorem \ref{thm:bspost}. 

The beta-Stacy bootstrap generalizes several Bayesian variants of the classical bootstrap: the \emph{Bayesian bootstrap} of \citet{Rubin1981}, the \emph{proper Bayesian boostrap} of \citet{Muliere1996}, and the Bayesian boostrap for censored data of \citet{Lo1993}. The first two assume that there is no censoring, while the last allows for censored data. Their relationships are summarized in Figure \ref{fig:1}.

Given uncensored observations $Y_1,\ldots,Y_n$, the \emph{Bayesian bootstrap} of \citet{Rubin1981} assigns $\phi(G)$ the same law as $\phi(\sum_{i=1}^n W_i I\{Y_n\leq x\})$, where $(W_1,\ldots,W_n)$ has as a uniform Dirichlet distribution (and thus it is an exchangeably weighted bootstrap; c.f. \citealp{Praestgaard1993}). Consequently, Rubin's bootstrap approximates the posterior law of $\phi(G)$ induced by the improper Dirichlet process $G\sim\textrm{DP}(0,F)$, i.e. the law of $\phi(G^{*})$, where $G^*\sim\textrm{DP}(n,n^{-1}\sum_{i=1}^n I\{Y_i\leq x\})$ \citep[Section 4.7]{Ghosal2017}. 

In contrast, the proper Bayesian bootstrap of \citet{Muliere1996} is defined according to a procedure akin to Algorithm \ref{def:bsb}. In detail, step 1 is the same, since $F^*=\widehat{F}$ (there is no censoring); in step 2, take $c(x)=k$ for all $x$; finally, step 3 and 4 are the same. Hence, when there is no censoring, the procedure of \citet{Muliere1996} is a special case of the beta-Stacy bootstrap (in general, neither is exchangeably weighted; c.f. \citealp{Praestgaard1993}). Their relation is illustrated in Figure \ref{fig:1} by arrow (a).

As a consequence, if there is no censoring the proper Bayesian bootstrap approximates (for large $m$) the posterior law of $\phi(G)$ induced by a proper Dirichlet process $G\sim\textrm{DP}(k,F)$ with $k>0$. More precisely, it approximates the law of $\phi(G^*)$ with $G^*\sim\textrm{DP}(k+n,\widehat{F})$ and $\widehat{F}=\frac{k}{k+n}F(x)+\frac{1}{k+n}\sum_{i=1}^nI\{T_i\leq x\}$. Thus, as $k\rightarrow 0$ (i.e. as the prior precision of the Dirichlet process vanishes), the proper Bayesian bootstrap will approximate the same posterior distribution as the procedure of \citet{Rubin1981}---c.f. \citet{Muliere1996}. This is illustrated by arrow (c) in Figure \ref{fig:1}.

Lo's procedure (\citeyear{Lo1993}) extends Rubin's bootstrap (\citeyear{Rubin1981}) to the case where censoring is possible---they coincide when there is no censoring; c.f. arrow (d) in Figure \ref{fig:1}. Specifically, the Lo's Bayesian bootstrap for censored data approximates the posterior law of $\phi(G)$ obtained from the improper beta-Stacy prior $\textrm{BS}(0,F)$ or, equivalently, an improper beta process (\citealp{Lo1993}). Better, Lo's bootstrap (\citeyear{Lo1993}) approximates the law of $\phi(G)$ with that of $\phi(G^*)$, where $G^*\sim\textrm{BS}(\widehat{c}(x),\widehat{G}(x))$, $\widehat{c}(x)=M(x)/\widehat{G}(x)$, and $\widehat{G}(x)$ is the Kaplan-Meier estimator (c.f. Section \ref{sec:bs}). This is the limit of the beta-Stacy posterior law from Theorem \ref{thm:bspost} as $c(x)\rightarrow 0$ for all $x>0$. Thus, Lo's procedure (\citeyear{Lo1993}) is obtained from ours in the limit of $c(x)\rightarrow 0$ for all $x>0$ (c.f. arrow (b) in Figure \ref{fig:1}).

In addition to the ones mentioned above, the beta-Stacy bootstrap also generalizes the Bayesian bootstrap for finite populations of \citet{Lo1988} and the P\`olya urn bootstrap of \citet{Muliere1998b}. These are obtained from the beta-Stacy bootstrap as previously done, assuming that $F$ is discrete and of finite support. 

\section{Generalization to the \texorpdfstring{$k$}{k}-sample case}

We now consider the setting where censored observations are available from $k$ independent groups. Specifically, we observe a sample time-to-event data $Y_{j,1},\ldots,Y_{j,n_j}$ generated by the cumulative distribution function $G_j$ for all $j=1,\ldots,k$. A similar setting arises, for example, in randomized trials with $k$ treatment arms and a survival end-point. Without loss of generality, we suppose that $k=2$.

In this setting, the the goal is often to compare summary measures of survival across groups. These correspond to joint functionals of the form $\phi(G_1,G_2)=f(G_1 h_1, \ldots, G_1 h_p, G_2 h_1, \ldots, G_2 h_p)$, where $f(x_1,\ldots,x_p,y_1,\ldots,y_p)$ and $h_1(x)$, $\ldots$, $h_p(x)$ are real-valued functions. Examples include the difference in expected survival times ($p=1$, $h_1(x)=x$, $f(x_1,y_1)=x_1-y_1$) or the ratio of survival probabilities ($p=1$, $h_1(x)=I\{x\geq t\}, f(x_1,y_1)=x_1/y_1$). Similarly as in Section \ref{sec:bsb}, we assume that $h_i$ is $F^*$-integrable and $F^*$-almost everywhere continuous, as well as that $f(x_1,\ldots,x_p,y_1,\ldots,y_p)$ is continuous.

If $G_1\sim\textrm{BS}(c_1,F_1)$ and $G_2\sim\textrm{BS}(c_2,F_2)$ independently,  we can use the beta-Stacy bootstrap to approximate the posterior law of $\phi(G_1,G_2)$ given the censored data $Y_{1,1},\ldots,Y_{1,n_1}$ and $Y_{2,1},\ldots,Y_{1,n_2}$. From Theorem \ref{thm:bspost}, this is the law of $\phi(G_1^*,G_2^*)$, where: $G_1^*$ and $G_2^*$ are independent; $G_j^*\sim\textrm{BS}(c_j^*,F_j^*)$ for each $j=1,2$; and $c_j^*$, $F_j^*$ are computed from the $j$-th group's data using Equations \ref{eqn:fpost}-\ref{eqn:cpost}. 

In more detail, let $G_{j,m}$ be the distribution function generated by  one iteration of the beta-Stacy boostrap in group $j=1,2$ (c.f. step 4 of Definition \ref{def:bsb}). Then, for large $m$, $\phi(G_{1,m}, G_{2,m})$ will be an approximate sample from the law of $\phi(G_1^*,G_2^*)$, as shown by the following proposition.

\begin{proposition}
$\phi(G_{1,m}, G_{2,m}) \rightarrow \phi(G_1^*,G_2^*)$ for $m\rightarrow+\infty$ conditional on $Y_{1,1}$, $\ldots$, $Y_{1,n_1}$, $Y_{2,1}$, $\ldots$, $Y_{1,n_2}$.
\end{proposition}
\begin{proof}
Since $G_{1,m}$ and $G_{2,m}$ are independent conditional on $Y_{1,1}$, $\ldots$, $Y_{1,n_1}$, $Y_{2,1}$, $\ldots$, $Y_{1,n_2}$, Corollary \ref{cor:1} implies that $(G_{1,m}h_1$, $\ldots$, $G_{1,m}h_p$, $G_{2,m}h_1$, $\ldots$, $G_{2,m}h_p)$ converges in law to $(G_1^*h_1$, $\ldots$, $G
_1^*h_p$, $G_2^*h_1$, $\ldots$, $G
_2^*h_p)$ as $m\rightarrow +\infty$. The thesis now follows from the continuous mapping theorem.
\end{proof}


\section{Implementing the beta-Stacy bootstrap}

To implement the beta-Stacy bootstrap, we use the following procedure to generate observations from $F^*$ (step 1 of Definition \ref{def:bsb}). To be concrete, we assume that $F$ is continuous (so $\Delta F(x)=0$ for all $x>0$) with density $f(x)$, but a similar method can also be used when $F$ is discrete.

Our approach is based on the relationship $F^*(x)=1-(1-F^*_d(x))(1-F^*_c(x))$ described in Section \ref{sec:bs}. This implies that if $X_d$ and $X_c$ are sampled independently from $F^*_d$ and $F^*_c$, respectively, then $X=\min(X_d,X_c)$ is a sample from $F^*$. We implement step 1 of Algorithm \ref{def:bsb} by iterating this process $m$ times. 

In detail, we sample $X_d$ from $F^*_d$ as follows. First, we note that, since $F$ is continuous, Equation \ref{eqn:disc} implies that $\Delta F^*_d(x)>0$ only if $\Delta N(x)>0$. Consequently, we can sample $X_d$ by defining it equal to $Y_j$ with probability $\Delta F^*_d(Y_j)$ for all $j=1,\ldots,n$, or $+\infty$ with probability $1-\sum_{j=1}^n \Delta F^*_d(Y_j)$. We do this using the inverse probability transform algorithm \citep[Chapter 3]{Robert2004}. 

Instead, we generate $X_c$ from $F^*_c$ in Equation \ref{eqn:fpost} using the inverse probability transform approach \citep[Chapter 3]{Robert2004}. Specifically, first we sample $U$ from the uniform distribution over $[0,1]$, then we define $X_c$ as the solution to the equation
$$\int _0^{X_c} \frac{c(x)f(x)}{c(x)(1-{F}(x))+M(x)}dx=-\log(1-U).$$
We approximate the above integral using Gaussian quadrature and compute $X_c$ using the bisection root-finding method \citep{Quarteroni2010}. 

\section{Empirical illustration}\label{sec:app}

We illustrate our procedure using survival data (freely available as part of the R dataset \texttt{survival::pbc}) from a randomized clinical trial of D-penicillamine for primary biliary cirrhosis of the liver \citep{Dickson1989}. In this trial, 312 cirrhosis patients were randomized to receive either D-penicilammine ($158$ patients) or placebo ($154$ patients). Patients in the D-penicilammine (respectively: placebo) arm accumulated a total of about $872$ ($842$) person-years of follow-up, during which $65$ ($60$) where observed. Overall, 187 (59.9\%) survival times were censored across study arms. Arm-specific Kaplan-Meier curves are shown in Figure \ref{fig:2}, panel a. 

Using these data, we compare the beta-Stacy bootstrap with another approach based on Algorithm a of \citet[Section 13.3.3]{Ghosal2017}---which we will denote as GvdVa. For any beta-Stacy process $G$, algorithm GvdVa can simulate approximate sample paths $\{G(x):x\in[0,T]\}$ over a prespecified bounded interval $[0,T]$. This algorithm is based on a discretization of $[0,T]$ by means of $N$ equally-spaced points, so that larger values of $N$ provide a better approximation to the beta-Stacy process (as we explain later, we use $N=5,000$ as reference in our analyses). We have chosen this algorithm as comparator because, compared to the others mentioned in the introduction, algorithm GvdVa is simpler to implement (like the beta-Stacy bootstrap, it is based on exact simulation steps and does not require sampling from unnormalized distributions; c.f. \citealp{DeBlasi2014}). Details are provided in the Supplementary Section S1. 

\begin{figure}[!th]
\centering
\includegraphics[scale=0.7]{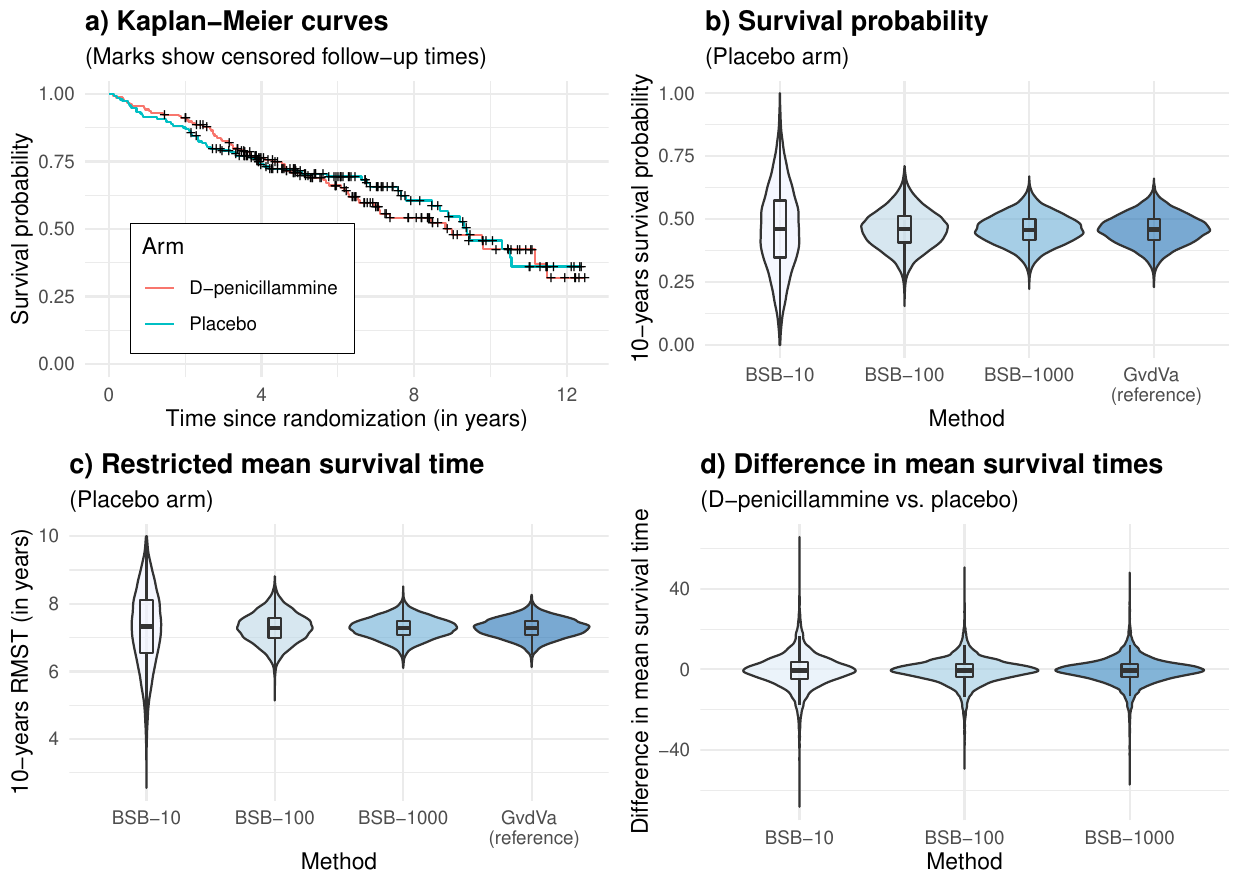}
\caption{Panel a: Kaplan-Meier curves for the Mayo clinic primary biliary chirrosis trial (c.f. Section \ref{sec:app}). Panels b and c: density estimates and box-plots of $10,000$ posterior samples of the 10-years survival probability (panel b) and the 10-years restricted mean-survival time (panel c) in the placebo arm; samples were obtained either with the beta-Stacy bootstrap (separately for $m=10$, $100$, and $1,000$) or using the reference GvdVa algorithm (c.f. Section \ref{sec:app}). Panel d: density estimates and box-plots of $10,000$ beta-Stacy bootstrap samples of the difference in mean survival times across arms (for $m=10$, $100$, and $1,000$).}
\label{fig:2}
\end{figure}

\subsection{Prior and posterior distributions}

Denote with $G_0$ and $G_1$ the cumulative distribution functions of survival times in the placebo and D-penicilammine arms, respectively. We assigned $G_i$ ($i=0,1$) an independent beta-Stacy prior $BS(c_i,F_i)$, where $F_i$ is the cumulative distribution function of an exponential random variable with median equal to 10 years. For simplicity, we assumed $c_i(x)=1$ for all $x\geq 0$. These prior distributions are fairly non-informative, since they are very diffuse around their expected values (Supplementary Figure S1). 

With these priors, the posterior means of $G_0$ and $G_1$ are practically indistinguishable from the corresponding Kaplan-Meier curves (Supplementary Figure S2). This is also confirmed by the Kolmogorov-Smirnov distances $D_i=\sup_{x\in [0,12]}|F_i^*(x)-\widehat{G}_i(x)|$  ($i=0,1)$, which compare the Kaplan-Meier estimate $\widehat{G}_i$ of $G_i$ and the corresponding posterior mean $F_i^*$ over the period from 0 to 12 years from randomization. We estimated that $D_0=0.004$ for the placebo arm, and $D_1=0.005$ for the D-penicilammine arm.

\subsection{Inference for single-sample summaries}\label{sec:res}

Using the beta-Stacy bootstrap and the GvdVa algorithm, we approximate the posterior distribution of two summaries of $G_0$: i) the 10-year survival probability in the placebo arm, i.e. $\phi_1(G_0)=1 - G_0(10) = G_0h$ with $h(x)=I\{x>10\}$, and ii) the 10-year restricted mean survival time in the placebo arm, i.e. $\phi_2(G_0)=\int_0^{10} [1-G_0(x)]dx = G_0 h'$ with $h'(x)=\min(x,10)$. 

In each case, we obtain 10,000 posterior samples. For the beta-Stacy bootstrap, we use $m=10$, $100$, and $1,000$, separately. To provide a reference against which to compare the beta-Stacy bootstrap, we implemented the GvdVa algorithm using a discretization of the time interval $[0,10]$ based on $N=5,000$ equally-spaced points. We chose by this value by iteratively increasing $N$ until the corresponding approximate posterior distributions of $\phi_1(G_0)$ and $\phi_2(G_0)$ seemed to stabilize (c.f. Supplementary Section S1). Note that algorithm GvdVa can be applied to $\phi_1(G_0)$ and $\phi_2(G_0)$ because they depend only on the values of $G_0(x)$ for $x\in[0, 10]$.

We use Kolmogor-Smirnov statistics to compare the distributions obtained from the beta-Stacy bootstrap and algorithm GvdVa. Specifically, for both summary measures $\phi_1(G^*_0)$ and $\phi_2(G^*_0)$ separately, we compute the statistics $\Delta_m = \sup_{x>0} |\widehat{F}_{1}(x)-\widehat{F}_{0,m}(x)|$, where: $m=10$, $100$, or $1,000$; $\widehat{F}_{0,m}(x)$ is the empirical distribution of the corresponding beta-Stacy bootstrap sample; and $\widehat{F}_{1}$ is the empirical distribution of the GvdVa samples.

Results are shown in Figures \ref{fig:2}b-c. For the 10-year survival probability (panel b), the distribution of beta-Stacy bootstrap samples approaches that obtained from algorithm GvdVa as $m$ increases. Indeed, the associated Kolmogorov-Smirnov statistics are $\Delta_{10}=0.24$, $\Delta_{100}=0.06$, and $\Delta_{1,000}=0.02$. Similar results where also obtained for the 10-year restricted mean survival (panel c), for which we computed $\Delta_{10}=0.32$, $\Delta_{100}=0.11$, and $\Delta_{1,000}=0.02$. The choice $m=1,000$ thus seems to have provided a good approximation to the posterior laws of interest.

\subsection{Difference in mean survival times}

We now consider the posterior law of the two-sample summary $\phi(G_1, G_0)=G_1 h - G_0 h$ defined by $h(x)=x$, i.e. the difference in mean survival times between the D-penicilammine arm and the placebo arm. In this case, it's hard to use the GvdVa algorithm to approximate the beta-Stacy posterior, because $h(x)$ has infinite support. On the contrary, we can still use the beta-Stacy bootstrap directly to generate approximate samples from the posterior law of $\phi(G_1, G_0)$.

In Figure 2d, we show the distribution of 10,000 posterior samples of the difference in mean survival times obtained with the beta-Stacy bootstrap, separately using $m=10$, $100$, or $1,000$. Compatibly with the previous results, the distribution of posterior samples stabilizes as $m$ increases. In particular, the density estimates and quartiles of the distributions for $m=100$ and $m=1,000$ are almost indistinguishable (the Kolmogorov-Smirnov distance between the two sample distributions was 0.007). These results again suggest that $m=1,000$ provided a good approximation to the relevant posterior distribution.

\subsection{Additional simulation study}

In Supplementary Section S2, we report a simulation study aimed at assessing how the proportion of censored sample observations may impact the beta-Stacy bootstrap approximation of the beta-Stacy posterior distribution. Results suggest that the proportion of censored observations does not impact the quality of the approximation in comparison to the reference GvdVa algorithm, provided that $m$ is sufficiently large. Compared to a scenario with no  censoring, simulation scenarios with higher censoring rates (up to 75\% of censored data) did not require larger values of $m$ to obtain the same quality of approximation ($m=1,000$ seemed to be acceptable in all considered scenarios).

\section{Concluding remarks}\label{sec:concl}

The beta-Stacy bootstrap is an algorithm to perform Bayesian non-parametric inference with censored data. This procedure generates approximate sampler a beta-Stacy process posterior \citep{Walker1997} without the need to tune Markov Chain Monte Carlo methods. The quality of the approximation is controlled by the number $m$ of samples from the posterior mean distribution (c.f. step 1 of Algorithm \ref{def:bsb}). Our simulations suggest that $m=1,000$ may generally provide a good approximation, independently of the proportion of event times that are affected by censoring. 

In place of the beta-Stacy process, many other non-parametric prior processes could  be used to estimate summaries $\phi(G)$ of the survival distribution function. Examples include piece-wise hazard processes \citep{Arjas1994}, the gamma and extended gamma processes \citep{Kalbfleisch1978, Dykstra1981}, P\`olya trees \citep{Mauldin1992, Muliere1997}, mixture models driven by random measures \citep{Kottas2006, RivaPalacio2021}, or Bayesian Additive Regression Trees \citep{Sparapani2016}. In comparison to the beta-Stacy process, computations using alternative prior processes may require the use of Markov Chain Monte Carlo samplers due to lack of conjugacy. Whether new Bayesian bootstraps could be derived for other conjugate processes (e.g \citealp{Muliere1997}) is a question for future research. 

Inference using the beta-Stacy process $BS(c, F)$ requires specification of both the precision function $c$ and the prior mean distribution function $F$. To avoid having to specify these in full, we might instead define $c_\theta$ and/or $F_\theta$ as a function of a scalar or multi-dimensional parameter $\theta$ (e.g. we might take $c(x)=\theta$ for all $x\geq 0$). Then, instead of specifying a single value of $\theta$, we could assign it a prior distribution $\pi(\theta)$. This approach leads to the specification of a mixture of beta-Stacy processes as the prior distribution for $G$, i.e. $G\sim \int BS(c_\theta, F_\theta)\pi(d\theta)$, in the same approach of \citet{Antoniak1974}. In future work, we will evaluate the use of the beta-Stacy bootstrap in Monte Carlo schemes for such mixtures  and their generalization for competing risks data \citep{Arfe2018a}.

\appendix
\section*{Appendix: technical lemmas and proofs}
\renewcommand{\thesection}{A} 

To prove Proposition \ref{prop:gdmconv}, we will use results related to convergence in law of random measures---c.f. \citet{Daley2007}, Section 11.1; see also \citet{Kallenberg2017}, Chapter 4. Let $W$ and $W_m$ be random measures over $[0,+\infty)$ that are finite on bounded intervals for every integer $m\geq 1$. Then, $W_m$ \emph{converges in law} to $W$ if and only if $W_mh \rightarrow Wh$ in law (as real-valued random variables) for every bounded continuous function $h:[0,+\infty)\rightarrow\mathbb{R}$ with bounded support. This happens if and only $\mathbb{E}[\exp(-W_mh)] \rightarrow \mathbb{E}[\exp(-Wh)]$ as $m\rightarrow +\infty$ for every such function $h$ \citep[Proposition 1.11.VIII]{Daley2007}.  

Let $D[0,+\infty)$ be the space of all right-continuous functions with left-hand limits with the Skorokhod topology \citep[Chapter VI, Section 1b]{Jacod2003}. The following result links convergence in law of  $W_m$ to $W$ to convergence in law of their cumulative distribution functions as random elements of $D[0,+\infty)$.

\begin{lemma}\label{lemma:skorokhod}
The random measure $W_m$ converges in law to $W$ if and only if the function $W_m(x)=W_m([0,x])$ converges in law to $W(x)=W([0,x])$ ($x\geq 0$) as random elements of $D[0,+\infty)$ with the Skorokhod topology. 
\end{lemma}
\begin{proof}
This result can be shown using a similar argument as that presented before Lemma 11.1.XI of \cite{Daley2007}.
\end{proof}

We now prove the following lemma, which implies that the measure $Z_m$ converges in law to $Z$ conditionally on $Y_1,\ldots,Y_n$, i.e. $\mathbb{E}[\exp(-Z_mh)|Y_1,\ldots,Y_n] \rightarrow \mathbb{E}[\exp(-Zh)|Y_1,\ldots,Y_n]$ as $m\rightarrow +\infty$ for every bounded continuous $h$ with bounded support. For simplicity of notation, let $\mathbb{E}_n[\cdot]=\mathbb{E}[\cdot|Y_1,\ldots,Y_n]$ be the conditional expectation with respect to $Y_1,\ldots,Y_n$. Also let $\mathbb{E}_{n,m}[\cdot]$ $=\mathbb{E}_n[\cdot|X_1,\ldots,X_m]$, where $X_1,\ldots,X_m$ are the variables from Step 1 of Algorithm \ref{def:bsb}. 

For any given $Y_1,\ldots, Y_n$ and $X_1,\ldots, X_m$, we define $\rho^*(x,u)$ (respectively: $\rho_m(x,u)$) as the function $\rho(x,u)$ in Equation \ref{eqn:rho}, but with $c^*$ and $F^*$ (respectively: $c^*$ and $F_m$) in place of $c$ and $F$. With these notations, by Lemma 1 of \citet{Ferguson1974} it is $-\log\mathbb{E}_n[\exp(-Z^*h)]=\int_{(0,+\infty)^2} (1-\exp(-uh(x))\rho^*(x,u)dF^*(x)du$ and, similarly,  
$-\log\mathbb{E}_{n,m}[\exp(-Z_m h)] = \int_{(0,+\infty)^2} $ $(1-\exp(-uh(x))\rho_m(x,u)dF_m(x)du$. 

\begin{lemma}\label{lemma:weakconv}
i) If $h:[0,+\infty)\rightarrow[0,+\infty)$ is a bounded measurable function with bounded support (but not necessarily continuous), then, conditionally on $Y_1,\ldots,Y_n$,  $Z_m h\rightarrow Z^*h$ in law as $m\rightarrow+\infty$. ii) The previous statement also holds for every bounded measurable $h:[0,+\infty)\rightarrow\mathbb{R}$ with bounded support. 
\end{lemma}
\begin{proof}
First we prove point (i). By dominated convergence, it suffices to show that $\mathbb{E}_{n,m}[\exp(-Z_mh)]$ $\rightarrow  \mathbb{E}_n[\exp(-Z^*h)]$ as $n\rightarrow+\infty$ with probability 1 for all functions $h$ such that $0\leq h(x) \leq H$ and $h(x)=0$ for all $x>l$ for some $H, l>0$. To do so, we note that $(1-e^{-uh(x)})\rho_m(x,u)\rightarrow (1-e^{-uh(x)})\rho^*(x,u)$ uniformly in $x$, and so $g_m(u)=\int_0^l (1-e^{-uh(x)})\rho_m(x,u)dF_m(x) \rightarrow g(u)=\int_0^l (1-e^{-uh(x)})\rho^*(x,u)dF^*(x)$, for all fixed $u$ with probability 1. This follows from the Glivenko-Cantelli theorem, the fact that $c^*(x)$ is bounded, and because the functions $x\mapsto e^{-x}$ and $r(x)$ are bounded and Lipschitz over $(0,+\infty)$. Now, fix $\delta>0$ such that $\overline{F}^*(l)>\delta>0$ (this is possible because $F(x)<1$ for all $x>0$). With probability 1, $\overline{F}_m(x)>\delta$ for all $x\leq l$ and large $m$. In such case, since $\epsilon\leq c(x) \leq \epsilon^{-1}$ and $1-\exp(-uh(x))\leq \min(uH,1)$, it is $g_m(u)\leq w(u)=\gamma\min(uH,1)\exp(-u\gamma\delta)/(1-e^{-u})$ for $u>0$ and some $\gamma>0$. Since $-\log\mathbb{E}_{n,m}[e^{-Z_mh}]=\int_0^{+\infty}g_m(u)du$, $-\log \mathbb{E}_n[e^{-Z^*h}]=\int_0^{+\infty}g(u)du$, and $\int_0^{+\infty}w(u)du<+\infty$, the thesis follows by dominated convergence.

To prove point (ii), let $h:[0,+\infty)\rightarrow\mathbb{R}$ be a bounded measurable function with bounded support. Define $h^+(x)=\max(0,h(x))$ and $h^-(x)=-\min(0, h(x))$, which are both bounded non-negative measurable functions with bounded support. Now, by point (i), $\mathbb{E}_n[\exp(-\lambda_1 G_m h^+ - \lambda_2G_m h^-)]$ $=$ $\mathbb{E}_n[\exp(- G_m (\lambda_1 h^+ +\lambda_2 h^-))]$ $\rightarrow$ $\mathbb{E}_n[\exp(-G^* (\lambda_1 h^+ +\lambda_2 h^-))]$ $=$ $\mathbb{E}_n[\exp(-\lambda_1 G^* h^+ - \lambda_2G^* h^-)]$ as $m\rightarrow +\infty$ for every $\lambda_1,\lambda_2\geq 0$. Consequently, $(G_m h^+,G_m h^-)\rightarrow (G^* h^+,G^* h^-)$ in law as a random vector by convergence of the corresponding joint Laplace transform \citep[Theorem 4.3]{Kallenberg1997}. Since $h=h^+-h^-$, the thesis follows from the continuous mapping theorem. \qedhere
\end{proof}

Using Lemmas \ref{lemma:skorokhod} and \ref{lemma:weakconv}, we can now prove that, conditionally on $Y_1,\ldots,Y_n$, $G_mh$ converges in distribution to $G^*h$ for every bounded continuous function $h$ with bounded support.

\begin{lemma}\label{lemma:Gweak}
$G_m$ converges in law to $G^*$ as $m\rightarrow+\infty$ conditionally on $Y_1,\ldots,Y_n$.
\end{lemma}
\begin{proof}
Let $\phi:D[0,+\infty)\rightarrow D[0,+\infty)$ be defined by $\phi(z(x))=1-\exp(-z(x))$ ($x\geq 0$) for every $z\in D[0,+\infty)$. Since the map $x\mapsto 1-\exp(-x)$ defined for every real $x\geq 0$ is Lipschitz-continuous, $\phi$ is also continuous with respect to the Skorokhod topology on $D[0,+\infty)$. Since $G_m(x)=\phi(Z_m(x))$ and $G^*(x)=\phi(Z^*(x))$ for every $x\geq 0$, the thesis now follows from Lemma \ref{lemma:skorokhod}, Lemma  \ref{lemma:weakconv} and the continuous mapping theorem. \qedhere

\end{proof}

We are now ready to prove Proposition \ref{prop:gdmconv}.

\begin{proof}[Proof of Proposition \ref{prop:gdmconv}]
From Lemma \ref{lemma:Gweak} and Proposition 4.19 of \citet{Kallenberg2017}, it follows that $G_mh\rightarrow Gh$ in law conditionally on $Y_1,\ldots,Y_n$ for every bounded continuous function $h$ (not necessarily with bounded support). Then, using an argument like the one in the proof of Lemma 4.12 of \citet{Kallenberg2017}, it follows that the same is true for every bounded measurable function $h$ (not necessarily continuous) such that $F^*(D_h)=0$. We now show that the thesis holds for any $F^*$-integrable $h$ (not necessarily bounded), provided that $F^*(D_h)=0$. In fact, by an argument like the one used to prove point (ii) of Lemma \ref{lemma:weakconv}, it suffices to show that this is true for any such non-negative $h$. 

Consequently, suppose that $h(x)\geq 0$ for every $x\geq 0$. By the Portmanteau theorem, it suffices to show that $|\mathbb{E}_n[f(G_mh)]-\mathbb{E}_n[f(G^*h)]|\rightarrow 0$ as $m\rightarrow+\infty$ for any real-valued function $f(x)$ such that $|f(x)|\leq K$ and $|f(x)-f(y)|\leq L|x-y|$ for some $K,L\geq 0$. To do this, let $M\geq 0$ and define $h_M(x) = \min(M, h(x))$. Then, $|\mathbb{E}_n[f(G_mh)]-\mathbb{E}_n[f(G^*h)]|\leq \Delta_1 + \Delta_2 + \Delta_3$, where $\Delta_1=\sup_m|\mathbb{E}_n[f(G_mh)]-\mathbb{E}_n[f(G_mh_M)]|$, $\Delta_2 = |\mathbb{E}_n[f(G_mh_M)]-\mathbb{E}_n[f(G^*h_M)]|$, and $\Delta_3 = |\mathbb{E}_n[f(G^*h_M)]-\mathbb{E}_n[f(G^*h)]|$. Now, $\Delta_2 \rightarrow 0$ for $m\rightarrow +\infty$, because $h_M$ is bounded, measurable, and $D_{h_M}\subseteq D_h$. Consequently, $\limsup_{m\rightarrow+\infty}\left|\mathbb{E}_n[f(G_mh)]-\mathbb{E}_n[f(G^*h)]\right|\leq \Delta_1 + \Delta_3$. Since $0\leq h_M(x) \leq h(x)$ for every $x\geq 0$ and $G^*\sim BS(c^*, F^*)$, we also have that $\Delta_3 \leq L\mathbb{E}_n[G^*(h-h_M)]=LF^*(h-h_M)$. By the Markov inequality, for every $\delta>0$ it holds that $\Delta_1 \leq \sup_m(\delta L + 2\delta^{-1}KL\mathbb{E}_n \left[ G_m (h-h_M)\right])=\delta L + 2\delta^{-1}KL F^*(h-h_M)$---where the last equality follows from $\mathbb{E}_n \left[ G_m (h-h_M)\right]=\mathbb{E}_n[\mathbb{E}_{n,m}\{G_m (h-h_M)\}]=\mathbb{E}_n[F_m(h-h_M)]$ (c.f. Section \ref{sec:bsb}). As a consequence, $\limsup_{m\rightarrow+\infty}\left|\mathbb{E}_n[f(G_mh)]-\mathbb{E}_n[f(G^*h)]\right|\leq \delta L + (L+2\delta^{-1}KL)F^*(h-h_M)$. However, by the dominated convergence theorem, $F^*(h-h_M)\rightarrow 0$ as $M\rightarrow+\infty$. Hence, the thesis follows by first letting $M\rightarrow+\infty$ and then $\delta\rightarrow 0$ from above.\qedhere

\end{proof}

\section*{Supplementary information}

Supplementary results and figures are available on-line at the publishers' website. This include supplementary Figures S1 and S2, a description of the GvdVa algorithm (Section S1), and additional simulation results based on varying censoring rates (Section S2). The code used to implement the beta-Stacy bootstrap and reproduce our results is available at \url{https://github.com/andreaarfe/}.

\section*{Acknowledgments}

We thank Alejandra Avalos-Pacheco, Massimiliano Russo, and Giovanni Parmigiani for their useful comments. Part of this work was developed while the first author was supported by a post-doctoral fellowship at the Harvard-MIT Center for Regulatory Science, Harvard Medical School. Analyses were conducted in R (version 4.1.2) using the libraries \texttt{mvQuad}, \texttt{Rcpp}, and \texttt{ggplot2}.

\title{Supplementary information for ``A general Bayesian bootstrap for censored data based on the beta-Stacy process''}
\author{Andrea Arf\`e\footnote{Department of Epidemiology and Biostatistics, Memorial Sloan Kettering Cancer Center, New York, NY 10017, United States. Website: \url{andreaarfe.wordpress.com}. E-mail: \url{arfea@mskcc.org}.}, Pietro Muliere\footnote{Department of Decision Sciences, Bocconi University, 20136 Milan, Italy. E-mail: \url{pietro.muliere@unibocconi.it}}}

\newpage
\renewcommand{\thesection}{S\arabic{section}}
\renewcommand{\thefigure}{S\arabic{figure}}
\setcounter{section}{0}

\section*{Supplementary material}
\setcounter{section}{0}

\section{Description of the GvdVa algorithm} 

We describe an approach to simulate approximate finite sample paths of a beta-Stacy process $G$ based on Algorithm a of Ghosal and van der Vaart \citep[Section 13.3.3]{Ghosal2017}. In Section 8 of the main manuscript, we refer to this method as the GvdVa algorithm, and we use it as a comparator for the beta-Stacy bootstrap. As in Section 8, let $G^*\sim BS(c^*,F^*)$, with $c^*$ and $F^*$ defined as in Theorem 3.1, and fix $T>0$ ($T=10$ in Section 8).

Algorithm GvdVa generates approximate sample paths $\{G^*(x): x\in[0,T]\}$ from the law of $G^*$ as follows. These can be used to approximate the posterior distribution of summaries $\phi(G)$, provided $\phi$ only depends on the values of $G(x)$ for $x\in[0,T]$. To do so, first, sample an approximate path $\{G^*(x): x\in[0,T]\}$, compute $\phi(G^*)$. This last value is a single sample from the law of $\phi(G)$ conditional on $Y_1,\ldots,Y_n$ (more can be obtained by iteration). 

To describe algorithm GvdVa, let $y_1<\ldots<y_d$ be the ordered points $y\in[0,T]$ where $\Delta F^*(y)>0$ (these correspond to the distinct values among $Y_1,\ldots,Y_n$ that are not censored and do not exceed $T$). Also let $z_i=(i-1)T/N$ for all $i=1,\ldots,N$, where $N$ is a positive integer. Then, algorithm GvdVa proceeds as follows:
\begin{enumerate}
\item Sample $U_j\sim \textrm{Beta}\left(c^*(y_j)\Delta F^*(y_j), c^*(y_j)\overline{F}^*(y_j)\right)$ for all $j=1,\ldots,d$. Here, $\Delta F^*(y_j)$ is computed as $\Delta F^*(y_j)=(F^*_d(y_j)-F^*_d(y_{j-1}))(1-F^*_c(y_j))$, where $y_0=0$ and $F^*_c$, $F^*_d$ are the continuous and discrete parts of $F^*$, respectively (c.f. Sections 3 and 7). 
\item Let $\alpha_i=c^*(z_i)[F_c^*(z_{i+1})-F_c^*(z_{i})]$ and $\beta_i=c^*(z_i)[1-F_c^*(z_{i+1})]$ for all $i=1,\ldots,N$.
\item Independently sample $V_i\sim\textrm{Beta}(\alpha_i,\beta_i)$ for all $i=1,\ldots,N$.
\item For all $x\in[0,T]$, define
$$G_N^*(x)=1-\prod_{y_j\leq x}(1-U_j)\prod_{z_i\leq x}(1-V_i).$$
\end{enumerate}

Assuming that $c$ and $F$ are continuous, as done in Sections 7 and 8, then Theorem 2 of \citet{Walker1997} implies that the stochastic process $(G_N^*(x))_{x\in[0,T]}$ defined by step 4 above converges in law to $(G^*(x))_{x\in[0,T]}$ as $N\rightarrow+\infty$ (i.e. as the grid of points $z_1,\ldots,z_N$ becomes increasingly fine) as a random element of the space $D[0,T]$ with the Skorokhod topology. As a consequence, for large $N$, the random function $G_N^*$ can be considered an approximate sample path of $G^*\sim BS(c^*,F^*)$ over $[0,T]$. 

In Section 8, we implemented the GvdVa algorithm with $N=5,000$ as a reference for the beta-Stacy bootstrap. We chose this value based on the generate approximate samples of the two summaries $\phi_1(G_0)$ (10-years survival probability) and $\phi_2(G_0)$ (10-years RMST) considered in Section 8.2. Specifically, we implemented algorithm GvdVa using increasing values of $N$ until the approximate posterior distributions of these two summaries seemed to stabilize. Results for $N=5, 50, 500, 5000$ are reported in the figure below. 

\includegraphics[scale=0.60]{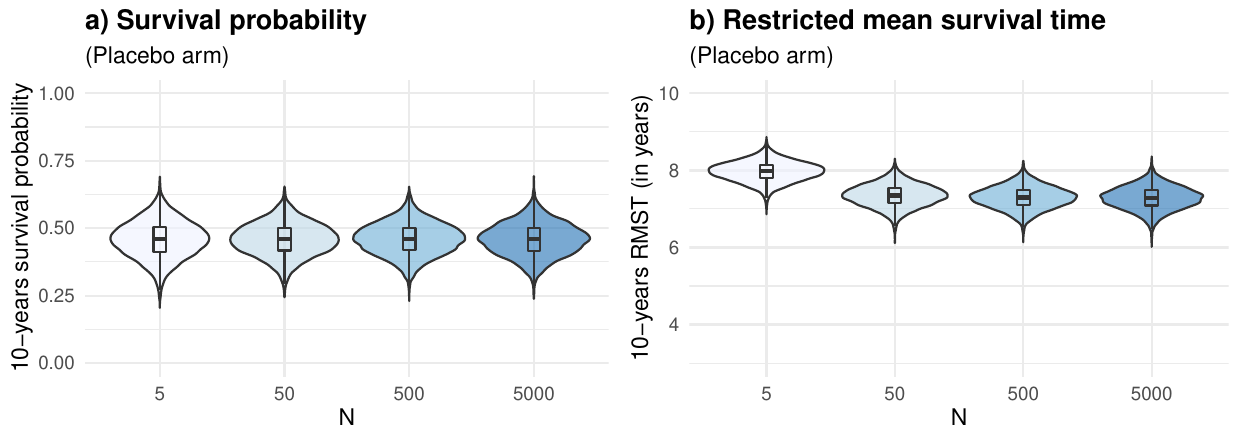}

\section{Additional simulation study}

We describe a simulation study aimed to asses whether the approximation to the beta-Stacy posterior of a summary $\phi(G)$ is impacted by the proportion of censored observations in the data. Specifically, we simulated survival according to scenarios defined by different values of the censoring probability $p_{cens}$ (i.e. the probability that a survival time is censored). For each scenario, we compared the empirical distribution of the posterior samples of $\phi_1(G_0)$ (10-years survival probability) and $\phi_2(G_0)$ (10-years RMST; c.f. Section 8.2) obtained with the beta-Stacy bootstrap (using $m=10 , 100$, or $1,000$) or the reference GvdVa algorithm (c.f. Section 8 and Supplementary Section S1 above).

In more detail, simulations were implemented according to the following steps. We fixed $\widehat{\lambda}_0 = 60/842$ and $n = 154$, the estimated event rate (deaths per person-year of follow-up) and number of patients in the placebo arm of the trial described in Section 8. Subsequently, we iterated the following steps 10,000 times separately for each value of $p_{cens}=0$ (i.e. no censoring), $0.25$, $0.5$, and $0.75$:
\begin{enumerate}
\item We simulated $T_1^c,\ldots,T_n^c$ independently from an exponential distribution with rate parameter equal to $\widehat{\lambda}_0$.
\item For each $i=1,\ldots,n$ independently, we set $\delta_i = 1$ with probability $p_{cens}$ or $\delta_i = 0$ with probability $1-p_{cens}$.
\item We formed the simulated dataset $Y_1=(T^c_1, \delta_1),\ldots, Y_n=(T_n^c, \delta_n)$ (c.f. the notation in Section 2).
\item We used the simulated dataset to generate a single approximate sample from each posterior laws of $\phi_1(G_0)$ and $\phi_2(G_0)$, using both the beta-Stacy bootstrap or the reference GvdVa algorithm. 
\item We stored the generated approximate posterior samples for analysis. 
\end{enumerate}

The figures below shows the empirical distribution of the posterior samples obtained with these steps. The results confirm that, as predicted from Proposition 4.1, the beta-Stacy bootstrap distribution will tend to that of the reference GvdVa algorithm for increasing $m$, independently of the proportion of censored observations. Regardless of the censoring rate, the Kolmogorov-Smirnov distance between the empirical distributions of beta-Stacy-bootstrap samples with $m=1,000$ and samples from the reference algorithm was less or equal than 0.03.

\begin{center}

\includegraphics[scale=0.70]{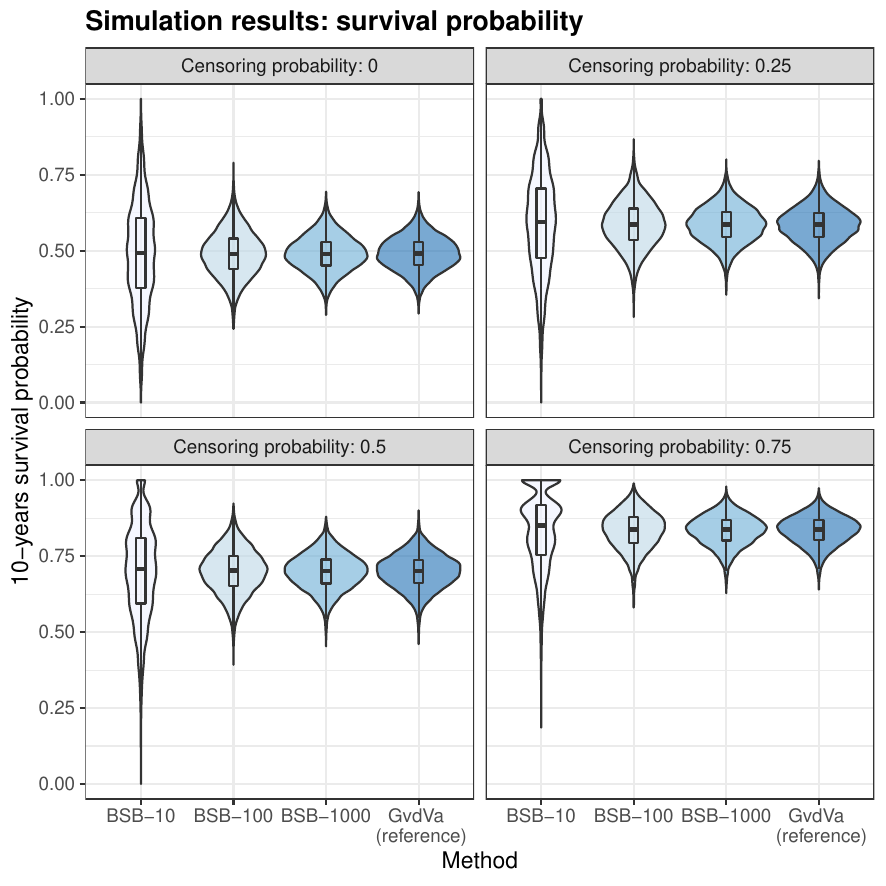}

\includegraphics[scale=0.70]{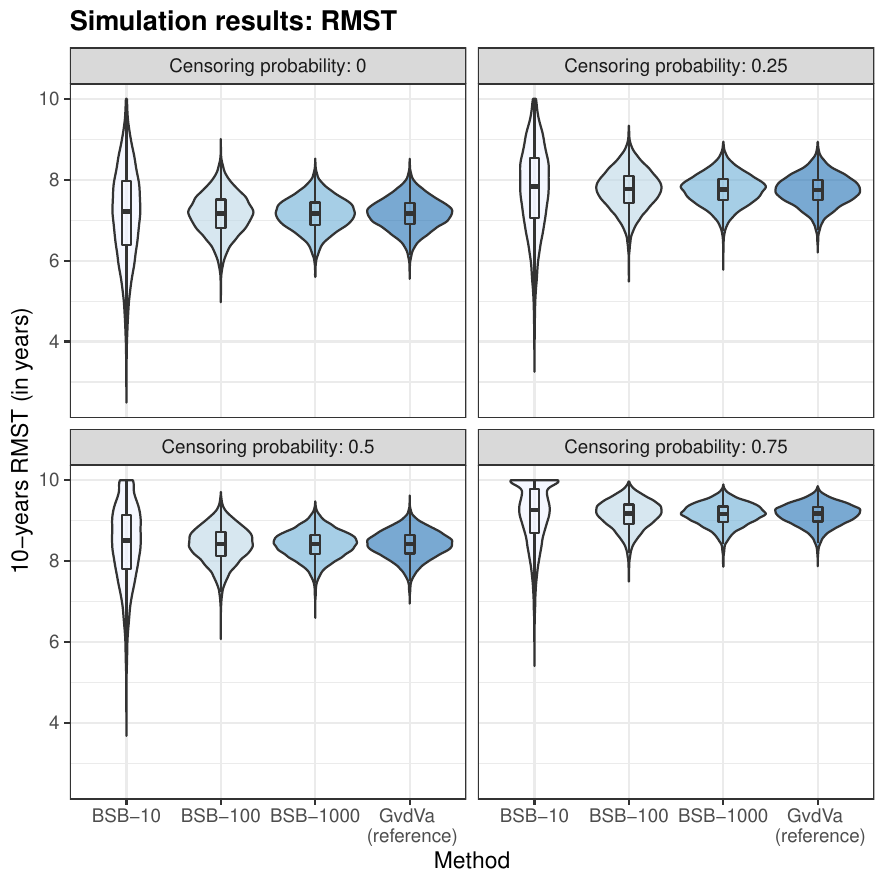}

\end{center}

\newpage

\section{Supplementary figures}

\begin{figure}[!th]
\centering
\caption{Comparison of the Kaplan-Meier curves (in orange) and beta-Stacy posterior means (in blue) obtained in the placebo (panel a) and D-penicilammine (panel b) arms of the Mayo Clinic's trial (c.f. Section 7).}
\includegraphics[scale=0.60]{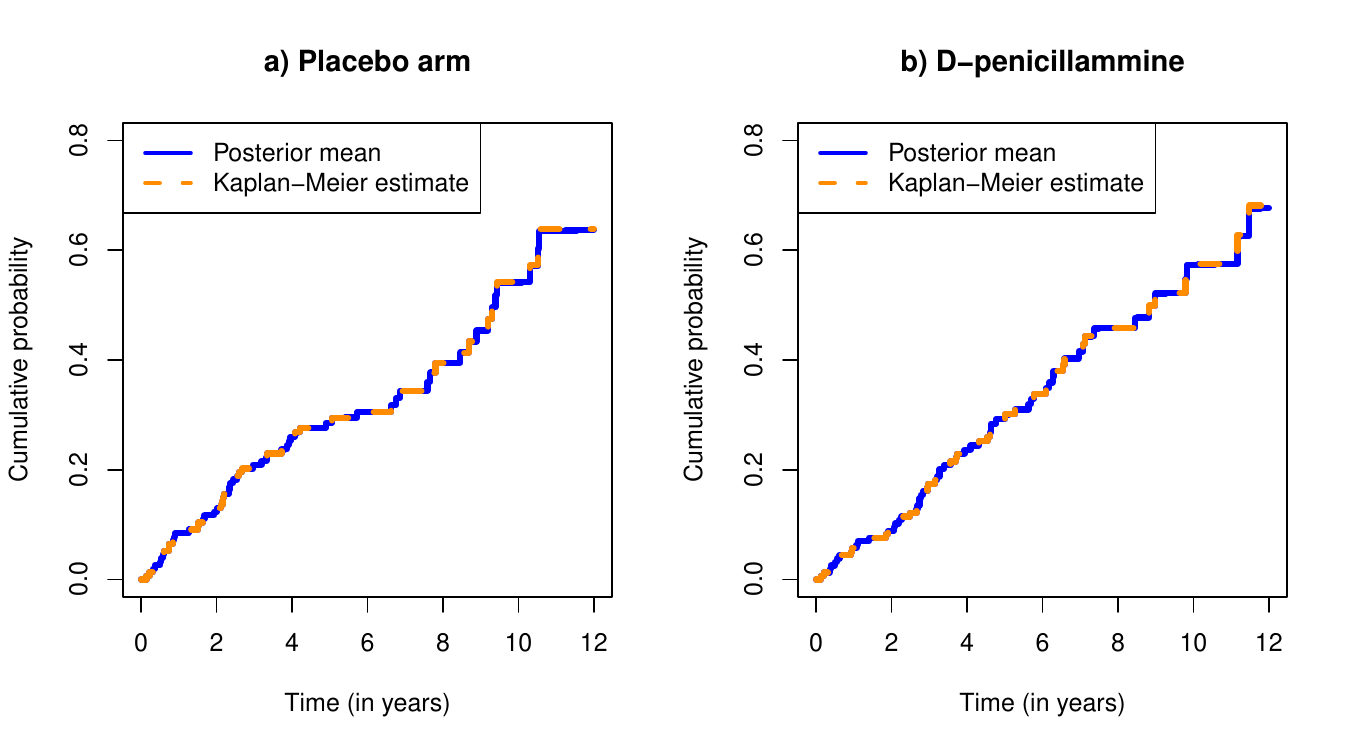}
\end{figure}

\begin{figure}[!th]
\centering
\caption{Plot of 100 survival curves (in black) sampled from the beta-Stacy prior distribution defined in Section 7. The red curve is the prior mean, i.e. the distribution function of an exponential random variable with median equal to 10 years. Samples were obtained using the approach described in Section S1 above.}
\includegraphics[scale=0.6]{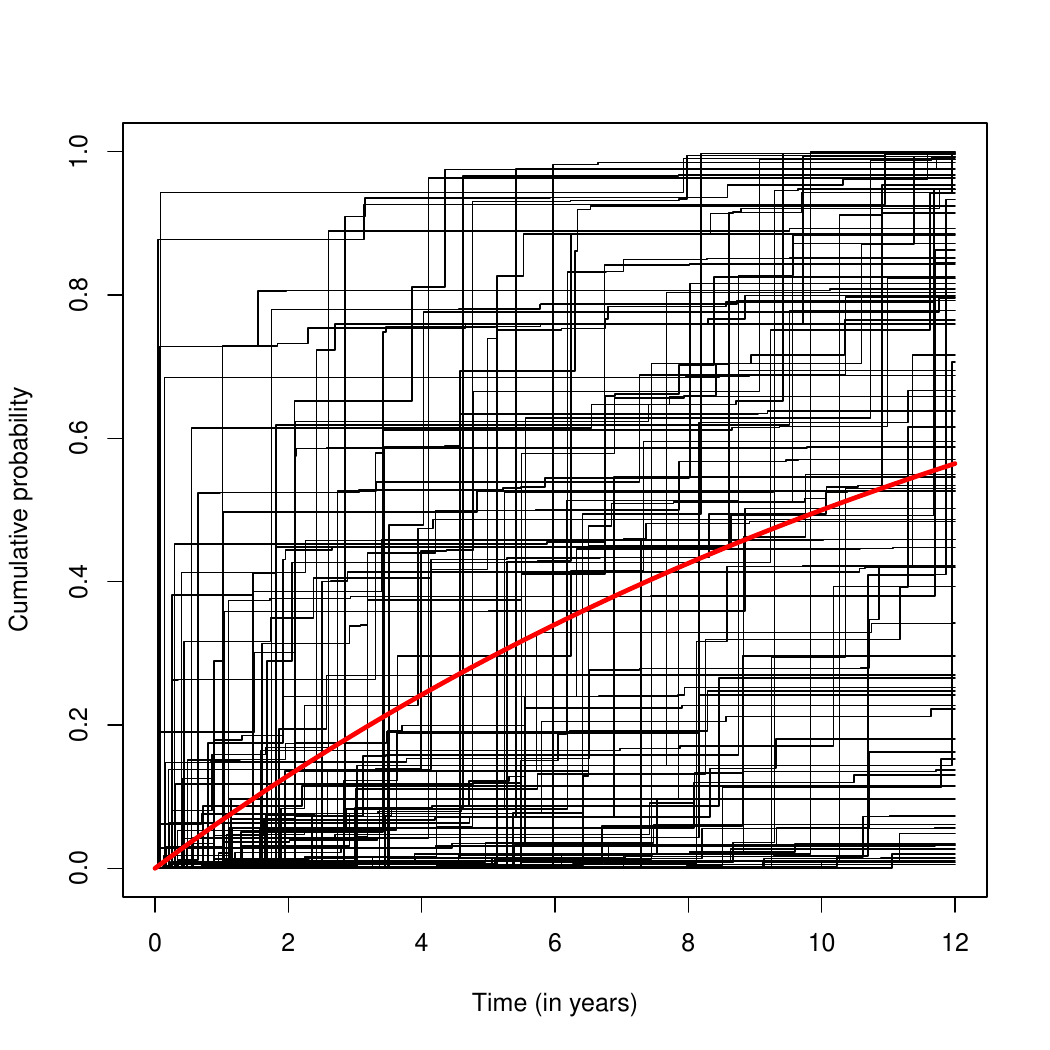}
\end{figure}

\newpage

\newpage

\bibliography{bibliography}
\bibliographystyle{rss}

\end{document}